\documentclass[a4paper,11pt]{article}
\usepackage{amsthm}
\usepackage[left=2cm,right=2cm,top=2cm,bottom=3cm]{geometry}
\usepackage{amsmath}
\usepackage{pxfonts}
\usepackage{authblk}
\date{\relax}
\author[*]{Ruediger Ehlers}
\author[**]{Salar Moarref}
\author[***]{Ufuk Topcu}
\affil[*]{University of Bremen and DFKI GmbH, Bremen, Germany}
\affil[**]{University of Pennsylvania, Philadelphia, PA, USA}
\affil[***]{University of Texas at Austin, TX, USA}

\usepackage[english]{babel}
\usepackage[utf8x]{inputenc}
\usepackage[colorlinks]{hyperref}
\usepackage{graphicx}
\usepackage{algorithm}
\usepackage{tikz}
\usetikzlibrary{positioning,angles,quotes}
\usepackage[noend]{algpseudocode}

\newtheorem{definition}{Definition}
\newtheorem{corollary}{Corollary}
\newtheorem{lemma}{Lemma}

\newtheorem{example}{Example}

\newtheorem{remark}{Remark}

\newcommand{\newterm}[1]{\emph{#1}}
\newcommand{\NN}{\mathbb{N}}

\newcommand{\TRUE}{\mathbf{true}}

\allowdisplaybreaks

\title{Risk-Averse $\omega$-regular Markov Decision Process Control\footnote{This is the ArXiV/CoRR preprint of the paper that was accepted  under the title ``{Risk-Averse Control of Markov Decision Processes with $\omega$-regular Objectives}'' at CDC 2016. The original publication is available on \href{http://dx.doi.org/10.1109/CDC.2016.7798306}{IEEE XPlore}.
\textcopyright{} 2016 IEEE. Personal use of this material is permitted. Permission from IEEE must be
obtained for all other uses, in any current or future media, including
reprinting/republishing this material for advertising or promotional purposes, creating new
collective works, for resale or redistribution to servers or lists, or reuse of any copyrighted
component of this work in other works.
}
}

\begin{document}
\maketitle

\begin{abstract}
Many control problems in environments that can be modeled as Markov decision processes (MDPs) concern infinite-time horizon specifications.
The classical aim in this context is to compute a control policy that maximizes the probability of satisfying the specification.
In many scenarios, there is however a non-zero probability of failure in every step of the system's execution. For infinite-time horizon specifications, this implies that the specification is violated with probability 1 in the long run no matter what policy is chosen, which prevents previous policy computation methods from being useful in these scenarios.

In this paper, we introduce a new optimization criterion for MDP policies that captures the task of working towards the satisfaction of some infinite-time horizon $\omega$-regular specification. The new criterion is applicable to MDPs in which the violation of the specification cannot be avoided in the long run. We give an algorithm to compute policies that are optimal in this criterion and show that it captures the ideas of \emph{optimism} and \emph{risk-averseness} in MDP control: while the computed policies are optimistic in that a MDP run enters a failure state relatively late, they are risk-averse by always maximizing the probability to reach their respective next goal state. We give results on two robot control scenarios to validate the usability of risk-averse MDP policies.
\end{abstract}

\section{Introduction}

The class of $\omega$-regular specifications allows to concisely capture long-term tasks for systems to be controlled.
Consequently, they have not only been used as specification formalism for the control of deterministic systems, but found applications in control of probabilistic systems. 
In the probabilistic case, the objective is typically to ensure that the specification holds almost surely or with the highest possible probability.

There are however many systems that do not admit any control strategy that satisfies an $\omega$-regular objective with a non-zero probability. In such a case, all controllers are equally bad: they violate the specification almost surely (or surely). 
If, for example, we have  a robot control scenario where there is always a small probability that the robot moves towards a wall (due to external influences), then a specification that forbids colliding with the wall cannot be fulfilled with a non-zero probability, as colliding with a wall almost surely eventually happens. 
Yet, researchers have proposed many approaches for controlling robots in such environments in practice. In a nutshell, these approaches are \emph{optimistic}: why should we be intimidated by events that are unavoidable but occur with small probability even in long time spans when we can still satisfy the specification for some time? 
{Such approaches are typically also \emph{risk-averse}: within the actions that are available to the robot, those that avoid the violation of the specification as long as possible are preferred.}
A reasonable strategy for the robot could, for example, try to stay clear of the walls and immediately take action when it happens to get closer to the wall at runtime. 
In this way, the robot could work towards satisfying its goals even though in the long run, it will eventually collide with a wall almost surely. 

On a theoretical level, the infinite-horizon nature of $\omega$-regular specifications however prevents the immediate application of \emph{optimism}, though. If with probability $1$, the specification is violated no matter what policy is used for controlling the system, then all control policies are equally bad and no best policy can be generated. While this fact advocates for an approach to system control that is not based on $\omega$-regular objectives, the infinitary nature of them allows to abstract from many details of the specification. As an example, we can state in the $\omega$-regular setting that the robot should visit each of two regions in a workspace infinitely often, which is a concise representation of the task of patrolling between these regions. 
The specification does not impose maximal times between visits to the regions, {which allows to optimize the risk-averseness of the policy}. %
Deviating from this concept would mean to impose time bounds between the visits to the regions. But then we get a tradeoff between optimizing for satisfying the specification as long as possible and the lengths of the patrolling periods. 
{
So it is desirable to keep the simplicity and conciseness of $\omega$-regular specifications to allow optimizing the probability to satisfy the specification for at least some time.

In this paper, we show how to compute \emph{optimistic}, yet \emph{risk-averse} policies for satisfying $\omega$-regular objectives in Markov decision processes (MDPs). We define an optimization criterion that captures the task of computing policies that satisfy $\omega$-regular control objective as long as possible, and give an algorithm to compute these policies.
The basic idea is that we require the policy to have a labeling that describes which states are considered to be \emph{goal states} by the policy, i.e., for which visiting them infinitely often ensures that the specification is satisfied. An \emph{optimally risk-averse policy} maximizes the probability for reaching the next goal state from the respective previous goal state. We argue that this criterion matches the intuitive idea that the controller should satisfy the specification as long as possible even if violation is almost surely unavoidable in the long run. %
We validate the usability of our risk-averse policy definition and the scalability of our policy computation algorithm on two case studies for robot control in probabilistic environments.
}

\section{Related Work}

MDPs are widely used in many areas such as engineering, economics and biology, and 
have been successfully used to model and control autonomous robots with uncertainty in their sensing and actuation (see e.g.,  \cite{ding2014optimal,lahijanian2012temporal,temizer2010collision,alterovitz2007stochastic}). In these domains, the behavior of the system cannot be predicted with certainty, but it can be modeled probabilistically through simulations or empirical trials. 
{Our results in this paper can be used in \emph{practical} settings in which the system cannot be controlled to satisfy a specification in the long run, but some amount of risk-taking is acceptable. }

MDPs are also referred to as $1\frac{1}{2}$-player games and belong to a broader class of \emph{stochastic games}.  
The algorithmic study of stochastic games with respect to $\omega$-regular objectives has recently attracted significant attention \cite{condon1992complexity,de1998concurrent,de2000concurrent,de2001quantitative,chatterjee2005complexity,chatterjee2006complexity}. 
See \cite{chatterjee2012survey} for a detailed survey. 
The central question about a game is whether a player has a strategy for winning the game. 
There are several definitions for \emph{winning} in a stochastic game \cite{chatterjee2012survey}. 
For example, one may ask if a player has a strategy that ensures a winning outcome of the game, no matter how the other player chooses her actions (\emph{sure winning}), or one may ask if a player has a strategy that achieves a winning outcome with probability $1$ (\emph{almost-sure winning}). 
In contrast to these \emph{qualitative} winning criteria, the \emph{quantitative} solution 
\cite{chatterjee2004quantitative,de2001quantitative} amounts to computing the value of the game, i.e., the maximal probability of winning that a player can guarantee against any strategy chosen by the opponent. 
{The choice of MDPs in this paper is motivated by their manageable complexity compared to more general classes of stochastic games, and by their applicability to many control problems. 
Ding et al.~\cite{DingEtAlMDPControl} gave an approach to compute MDP policies that maximize the probability of satisfying an $\omega$-regular specification. They applied their algorithm to robot indoor navigation. Svorenova et al.~\cite{DBLP:conf/cdc/SvorenovaCB13} considered the problem of minimizing the expected cost in between reaching \emph{goal states} in MDPs for $\omega$-regular specifications. Our work uses a similar notion of goal states.
}
None of the mentioned works consider the syn\-the\-sis of risk-averse policies in case there is no strategy that wins with a probability of greater than $0$.

\section{Preliminaries}

\paragraph{MDPs}
A \newterm{Markov decision process} is defined as a tuple $\mathcal{M} = (S,A,\Sigma,P,L,s_0)$, where $S$ is a finite set of \newterm{states}, $A$ is a finite set of \newterm{actions}, $\Sigma$ is the \newterm{label alphabet}, $P : S \times A \rightarrow \mathcal{P}(S) \cup \{\bot\}$ is the \newterm{transition function}, where $\mathcal{P}(S)$ denotes the probability distributions over $S$, $L : S \rightarrow \Sigma$ is the labeling function of $\mathcal{M}$, and $s_0 \in S$ is the initial state of the Markov chain. We say that some finite sequence $\pi = \pi_0 \ldots \pi_n \in S^*$ is a \newterm{finite trace} (or \newterm{run}) of $\mathcal{M}$ if there exists a sequence of actions $\rho = \rho_0 \ldots \rho_{n-1} \in A^*$ such that for all $i \in \{0, \ldots, n-1\}$, we have $P(\pi_i,\rho_i) \neq \bot$ and $P(\pi_i,\rho_i)(\pi_{i+1}) > 0$. We say that the combined probability of $(\pi,\rho)$ is  $\prod_{i=0}^{n-1} P(\pi_i,\rho_i)(\pi_{i+1})$. The definition of finite traces carries over to infinite traces.

A \newterm{Markov chain} is a Markov decision process (MDP) in which $A = \{\cdot\}$. A Markov chain introduces the usual probability measure over sets of infinite traces.

A \newterm{policy} for an MDP is a function $f : S^* \rightarrow \mathcal{P}(A)$ such that for all $s_0 \ldots s_n \in S^*$, we have $f(s_0 \ldots s_n)(a) = 0$ for all actions $a$ such that $P(s_n,a) = \bot$. A policy induces an infinite-state Markov chain $\mathcal{C}' = (S',\{\cdot\},\Sigma,P',L',s_0)$ with $S' = S^*$, $L'(t_0 \ldots t_n) = L(t_n)$ for all $t_0 \ldots t_n \in S'$, and for all $t_0 \ldots t_n, u_0 \ldots u_m \in S'$, we have $P'(t_0 \ldots t_n,\cdot)(u_0 \ldots u_m) = \sum_{a \in A} P(t_n,a) \cdot f(t_0 \ldots t_n)(a)$ if $u_0 \ldots u_{m-1} = t_0 \ldots t_n$, and $P'(t_0 \ldots t_n, \cdot)(u_0 \ldots u_m) = 0$ otherwise.

Policies for MDPs can be \newterm{positional} or \newterm{finite-state}. For a positional policy, for all $\pi = \pi_0 \ldots \pi_n \in S^*$ and $\pi' = \pi'_0 \ldots \pi'_m \in S^*$, we have that $f(\pi)=f(\pi')$ if $\pi_n = \pi'_m$. For a finite-state strategy, there exists a finite-state automaton $\mathcal{F} = (Q,S,\delta,q_0)$ with $Q$ being a finite set of states, $q_0 \in Q$, and $\delta : Q \times S \rightarrow Q$ such that there is a function $f' : Q \rightarrow \mathcal{P}(A)$ such that for all $\pi = \pi_0 \ldots \pi_n \in S^*$, we have that $f(\pi) = f'(q)$ for $q = \delta(\ldots \delta(\delta(q_0,\pi_0),\pi_1), \ldots, \pi_n)$.

{In literature, MDPs often also have a \emph{reward function}. As in some other work on $\omega$-regular MDP control \cite{DingEtAlMDPControl}, we do not need it in this paper and have thus omitted the reward function in the MDP definition.}
An MDP can be represented graphically by drawing the states as nodes in a graph, marking the initial state and letting the transitions be represented by groups of \newterm{edges}, which are in turn labeled by their transition probabilities. The groups of edges are labeled by their actions. Disallowed actions, i.e., for which we have $P(s,a)=\bot$, are not shown. %

\paragraph{Parity automata and $\omega$-specifications}
Given an alphabet $\Sigma$, an $\omega$-regular specification is a subset of $\Sigma^\omega$ that is representable as the language of a \newterm{deterministic parity word automaton}. These automata are defined as tuples $\mathcal{A} = (Q,\Sigma,\delta,q_0,C)$, where $Q$ is a finite set of states, $\Sigma$ is an alphabet, $\delta : Q \times \Sigma \rightarrow Q$ is the transition function of $\mathcal{A}$, $q_0 \in Q$ is the initial state of the automaton, and $C : Q \rightarrow \NN$ is the \newterm{coloring function}. Given a word $w = w_0 w_1 \ldots \in \Sigma^\omega$, $\mathcal{A}$ induces a \newterm{trace} $\pi = \pi_0 \pi_1 \ldots \in Q^\omega$ such that for all $i \in \mathbb{N}$, we have $\pi_{i+1} = \delta(\pi_i,w_i)$. Let $\mathsf{inf}$ be a function that maps an infinite sequence onto the elements of the sequence that occur infinitely often in it. A trace $\pi$ of $\mathcal{A}$ is called \newterm{accepting} if $\max(\mathsf{inf}( c(\pi_0) c(\pi_1) c(\pi_2) \ldots))$ is even. An automaton is said to accept a word $w$ if there exists an accepting trace for it. The set of all words accepted by the automaton is called its \newterm{language}.

\paragraph{Reachability MDPs}
A reachability MDP $\mathcal{M} = (S,A,\Sigma,P,L,s_0,g)$ consists of the usual MDP elements plus a function $g : S \rightarrow \{0,1\}$, which assigns to every state $s \in S$ either $0$ or $1$ depending on whether it is a \newterm{goal state} or not. A policy $f$ for $\mathcal{M}$ induces for every state $s$ a \newterm{value} $v(s) \in [0,1]$ that states the probability measure of the traces starting in $s$ and visiting a state $s' \in S$ with $g(s')=1$ when executing the policy, i.e., in the Markov chain induced by $\mathcal{M}$ and $f$ starting from state $s$. A policy that maximizes the values from all starting states is called \newterm{optimal} and it is known that in reachability MDPs, positional optimal policies exist \cite{condon1992complexity}. The values of the states induced by an optimal policy are also called the \newterm{state values} of the reachability MDP. These can be computed either by \newterm{policy iteration} or \newterm{value iteration} algorithms \cite{Sigaud:2010:MDP:1841781}.
In the latter case, a sequence of vectors $\vec x_1, \vec x_2, \ldots \in [0,1]^{|S|}$ is computed such that for every $i \in \NN$, $x_{i+1}$ is closer to the vector of state values than $x_i$. Value iteration is normally programmed to abort computation if at some point, $||x_{i+1} - x_i|| \leq \epsilon$ for some value $\epsilon$ and some norm $||\cdot||$. When starting with $\vec x_0$ being equivalent to $g$, the approximations are all under-approximations of the actual state values (modulo rounding errors).

\section{Problem Definition}

\begin{definition}[Parity MDP]
The product of an MDP $\mathcal{M} = (S,A,\Sigma,P,L,s_0)$ and a parity word automaton $\mathcal{A} = (Q,\Sigma,\delta,q_0,C)$ is  an MDP $\mathcal{M}' = (S', A, \Sigma, P', C', s'_0)$ with a coloring function instead of a labeling function where:
\begin{align*}
S' & = S \times Q, \\
s'_0 & = (s_0,q_0), \\
C'(s,q) & = C(q) \text{ for all } (s,q) \in S', \text{and} \\
P'((s,q),a)((s',q')) & = \begin{cases} P'(s,a)(s') & \text{if } q' = \delta(q,L(s')) \\
0 & \text{else} \end{cases} \\
& \quad \text{for all } (s,q), (s',q') \in S', a \in A.
\end{align*}
An infinite trace $\pi_0 \pi_1 \ldots \in S'^\omega$ in $\mathcal{M}'$ is said to be \newterm{accepting} if the highest number occurring infinitely often in the sequence $C'(\pi_0) C'(\pi_1) \ldots$ is even.
\end{definition}

{
A parity MDP captures a control problem in a probabilistic environment. We say that some trace $\pi = \pi_0 \pi_1 \ldots$ (or run) of the MDP is \newterm{accepting} if the trace fulfills the parity acceptance condition defined in the $C$ component in the MDP. Let us consider an example.

}

\begin{example}
\label{example:motivation}
As a first example, we consider a simple robot with unicycle dynamics in a two-dimensional gridded world. The workspace, which we depict in Figure~\ref
{fig:simpleRobot}, has 70$\times$40 cells and the robot always has one out of eight possible current directions. The speed of the robot is constant, and it needs to avoid hitting the workspace boundaries or the static obstacles. In order to model the scenario as an MDP, we use a semantics with a fixed time step. We shift the current cell into the current direction of travel by 2 cells, extend the resulting rectangle by $0.1$ into every direction to account for imprecise motion, and then assign transition probabilities that are proportional to the overlap of the rectangle with the world cells. There is an additional special error state in the MDP that represents crashes. In every step of the execution, the policy can decide to increase or decrease the current direction by $1$ step (out of $8$). This turning operation may fail with a probability of $0.2$ - in the case of failure, the direction of the robot is not changed. The MDP has $70 \cdot 40 \cdot 8 + 1 = 22401$ states, $67201$ state/action pairs, and $681591$ \newterm{edges}, i.e., pairs $(s,a,s')$ in the MDP $\mathcal{M} = (S,A,\Sigma,P,L,s_0)$ with $P(s,a)(s')>0$.

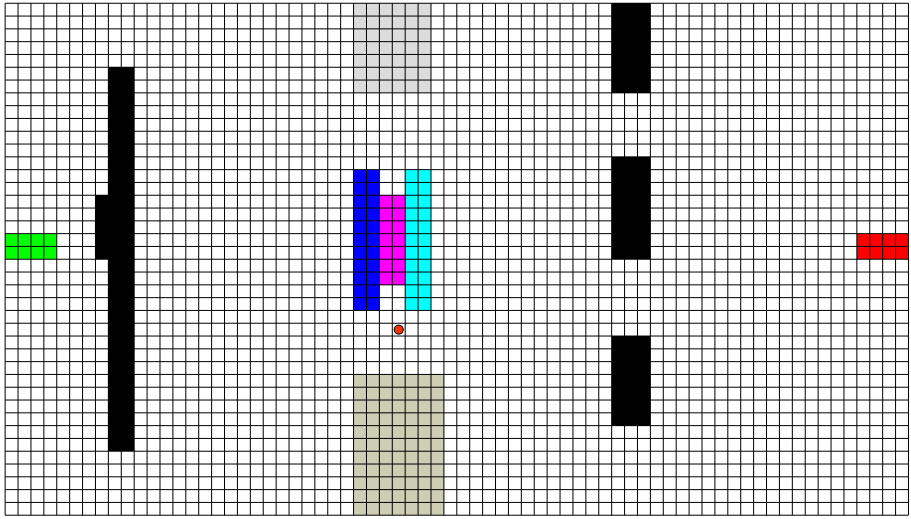
\begin{figure}
\centering\resizebox{0.7\columnwidth}{!}{%
\begin{tikzpicture}[yscale=-1]
\definecolor{mycolor0}{RGB}{255,255,255}
\definecolor{mycolor1}{RGB}{0,0,0}
\definecolor{mycolor2}{RGB}{0,255,0}
\definecolor{mycolor3}{RGB}{255,0,1}
\definecolor{mycolor4}{RGB}{0,0,255}
\definecolor{mycolor5}{RGB}{255,0,255}
\definecolor{mycolor6}{RGB}{0,255,255}
\definecolor{mycolor7}{RGB}{220,220,220}
\definecolor{mycolor8}{RGB}{205,206,179}
\path[fill=mycolor0] (0,0) rectangle (27,1);
\path[fill=mycolor7] (27,0) rectangle (33,1);
\path[fill=mycolor0] (33,0) rectangle (47,1);
\path[fill=mycolor1] (47,0) rectangle (50,1);
\path[fill=mycolor0] (50,0) rectangle (70,1);
\path[fill=mycolor0] (0,1) rectangle (27,2);
\path[fill=mycolor7] (27,1) rectangle (33,2);
\path[fill=mycolor0] (33,1) rectangle (47,2);
\path[fill=mycolor1] (47,1) rectangle (50,2);
\path[fill=mycolor0] (50,1) rectangle (70,2);
\path[fill=mycolor0] (0,2) rectangle (27,3);
\path[fill=mycolor7] (27,2) rectangle (33,3);
\path[fill=mycolor0] (33,2) rectangle (47,3);
\path[fill=mycolor1] (47,2) rectangle (50,3);
\path[fill=mycolor0] (50,2) rectangle (70,3);
\path[fill=mycolor0] (0,3) rectangle (27,4);
\path[fill=mycolor7] (27,3) rectangle (33,4);
\path[fill=mycolor0] (33,3) rectangle (47,4);
\path[fill=mycolor1] (47,3) rectangle (50,4);
\path[fill=mycolor0] (50,3) rectangle (70,4);
\path[fill=mycolor0] (0,4) rectangle (27,5);
\path[fill=mycolor7] (27,4) rectangle (33,5);
\path[fill=mycolor0] (33,4) rectangle (47,5);
\path[fill=mycolor1] (47,4) rectangle (50,5);
\path[fill=mycolor0] (50,4) rectangle (70,5);
\path[fill=mycolor0] (0,5) rectangle (8,6);
\path[fill=mycolor1] (8,5) rectangle (10,6);
\path[fill=mycolor0] (10,5) rectangle (27,6);
\path[fill=mycolor7] (27,5) rectangle (33,6);
\path[fill=mycolor0] (33,5) rectangle (47,6);
\path[fill=mycolor1] (47,5) rectangle (50,6);
\path[fill=mycolor0] (50,5) rectangle (70,6);
\path[fill=mycolor0] (0,6) rectangle (8,7);
\path[fill=mycolor1] (8,6) rectangle (10,7);
\path[fill=mycolor0] (10,6) rectangle (27,7);
\path[fill=mycolor7] (27,6) rectangle (33,7);
\path[fill=mycolor0] (33,6) rectangle (47,7);
\path[fill=mycolor1] (47,6) rectangle (50,7);
\path[fill=mycolor0] (50,6) rectangle (70,7);
\path[fill=mycolor0] (0,7) rectangle (8,8);
\path[fill=mycolor1] (8,7) rectangle (10,8);
\path[fill=mycolor0] (10,7) rectangle (70,8);
\path[fill=mycolor0] (0,8) rectangle (8,9);
\path[fill=mycolor1] (8,8) rectangle (10,9);
\path[fill=mycolor0] (10,8) rectangle (70,9);
\path[fill=mycolor0] (0,9) rectangle (8,10);
\path[fill=mycolor1] (8,9) rectangle (10,10);
\path[fill=mycolor0] (10,9) rectangle (70,10);
\path[fill=mycolor0] (0,10) rectangle (8,11);
\path[fill=mycolor1] (8,10) rectangle (10,11);
\path[fill=mycolor0] (10,10) rectangle (70,11);
\path[fill=mycolor0] (0,11) rectangle (8,12);
\path[fill=mycolor1] (8,11) rectangle (10,12);
\path[fill=mycolor0] (10,11) rectangle (70,12);
\path[fill=mycolor0] (0,12) rectangle (8,13);
\path[fill=mycolor1] (8,12) rectangle (10,13);
\path[fill=mycolor0] (10,12) rectangle (47,13);
\path[fill=mycolor1] (47,12) rectangle (50,13);
\path[fill=mycolor0] (50,12) rectangle (70,13);
\path[fill=mycolor0] (0,13) rectangle (8,14);
\path[fill=mycolor1] (8,13) rectangle (10,14);
\path[fill=mycolor0] (10,13) rectangle (27,14);
\path[fill=mycolor4] (27,13) rectangle (29,14);
\path[fill=mycolor0] (29,13) rectangle (31,14);
\path[fill=mycolor6] (31,13) rectangle (33,14);
\path[fill=mycolor0] (33,13) rectangle (47,14);
\path[fill=mycolor1] (47,13) rectangle (50,14);
\path[fill=mycolor0] (50,13) rectangle (70,14);
\path[fill=mycolor0] (0,14) rectangle (8,15);
\path[fill=mycolor1] (8,14) rectangle (10,15);
\path[fill=mycolor0] (10,14) rectangle (27,15);
\path[fill=mycolor4] (27,14) rectangle (29,15);
\path[fill=mycolor0] (29,14) rectangle (31,15);
\path[fill=mycolor6] (31,14) rectangle (33,15);
\path[fill=mycolor0] (33,14) rectangle (47,15);
\path[fill=mycolor1] (47,14) rectangle (50,15);
\path[fill=mycolor0] (50,14) rectangle (70,15);
\path[fill=mycolor0] (0,15) rectangle (7,16);
\path[fill=mycolor1] (7,15) rectangle (10,16);
\path[fill=mycolor0] (10,15) rectangle (27,16);
\path[fill=mycolor4] (27,15) rectangle (29,16);
\path[fill=mycolor5] (29,15) rectangle (31,16);
\path[fill=mycolor6] (31,15) rectangle (33,16);
\path[fill=mycolor0] (33,15) rectangle (47,16);
\path[fill=mycolor1] (47,15) rectangle (50,16);
\path[fill=mycolor0] (50,15) rectangle (70,16);
\path[fill=mycolor0] (0,16) rectangle (7,17);
\path[fill=mycolor1] (7,16) rectangle (10,17);
\path[fill=mycolor0] (10,16) rectangle (27,17);
\path[fill=mycolor4] (27,16) rectangle (29,17);
\path[fill=mycolor5] (29,16) rectangle (31,17);
\path[fill=mycolor6] (31,16) rectangle (33,17);
\path[fill=mycolor0] (33,16) rectangle (47,17);
\path[fill=mycolor1] (47,16) rectangle (50,17);
\path[fill=mycolor0] (50,16) rectangle (70,17);
\path[fill=mycolor0] (0,17) rectangle (7,18);
\path[fill=mycolor1] (7,17) rectangle (10,18);
\path[fill=mycolor0] (10,17) rectangle (27,18);
\path[fill=mycolor4] (27,17) rectangle (29,18);
\path[fill=mycolor5] (29,17) rectangle (31,18);
\path[fill=mycolor6] (31,17) rectangle (33,18);
\path[fill=mycolor0] (33,17) rectangle (47,18);
\path[fill=mycolor1] (47,17) rectangle (50,18);
\path[fill=mycolor0] (50,17) rectangle (70,18);
\path[fill=mycolor2] (0,18) rectangle (4,19);
\path[fill=mycolor0] (4,18) rectangle (7,19);
\path[fill=mycolor1] (7,18) rectangle (10,19);
\path[fill=mycolor0] (10,18) rectangle (27,19);
\path[fill=mycolor4] (27,18) rectangle (29,19);
\path[fill=mycolor5] (29,18) rectangle (31,19);
\path[fill=mycolor6] (31,18) rectangle (33,19);
\path[fill=mycolor0] (33,18) rectangle (47,19);
\path[fill=mycolor1] (47,18) rectangle (50,19);
\path[fill=mycolor0] (50,18) rectangle (66,19);
\path[fill=mycolor3] (66,18) rectangle (70,19);
\path[fill=mycolor2] (0,19) rectangle (4,20);
\path[fill=mycolor0] (4,19) rectangle (7,20);
\path[fill=mycolor1] (7,19) rectangle (10,20);
\path[fill=mycolor0] (10,19) rectangle (27,20);
\path[fill=mycolor4] (27,19) rectangle (29,20);
\path[fill=mycolor5] (29,19) rectangle (31,20);
\path[fill=mycolor6] (31,19) rectangle (33,20);
\path[fill=mycolor0] (33,19) rectangle (47,20);
\path[fill=mycolor1] (47,19) rectangle (50,20);
\path[fill=mycolor0] (50,19) rectangle (66,20);
\path[fill=mycolor3] (66,19) rectangle (70,20);
\path[fill=mycolor0] (0,20) rectangle (8,21);
\path[fill=mycolor1] (8,20) rectangle (10,21);
\path[fill=mycolor0] (10,20) rectangle (27,21);
\path[fill=mycolor4] (27,20) rectangle (29,21);
\path[fill=mycolor5] (29,20) rectangle (31,21);
\path[fill=mycolor6] (31,20) rectangle (33,21);
\path[fill=mycolor0] (33,20) rectangle (70,21);
\path[fill=mycolor0] (0,21) rectangle (8,22);
\path[fill=mycolor1] (8,21) rectangle (10,22);
\path[fill=mycolor0] (10,21) rectangle (27,22);
\path[fill=mycolor4] (27,21) rectangle (29,22);
\path[fill=mycolor5] (29,21) rectangle (31,22);
\path[fill=mycolor6] (31,21) rectangle (33,22);
\path[fill=mycolor0] (33,21) rectangle (70,22);
\path[fill=mycolor0] (0,22) rectangle (8,23);
\path[fill=mycolor1] (8,22) rectangle (10,23);
\path[fill=mycolor0] (10,22) rectangle (27,23);
\path[fill=mycolor4] (27,22) rectangle (29,23);
\path[fill=mycolor0] (29,22) rectangle (31,23);
\path[fill=mycolor6] (31,22) rectangle (33,23);
\path[fill=mycolor0] (33,22) rectangle (70,23);
\path[fill=mycolor0] (0,23) rectangle (8,24);
\path[fill=mycolor1] (8,23) rectangle (10,24);
\path[fill=mycolor0] (10,23) rectangle (27,24);
\path[fill=mycolor4] (27,23) rectangle (29,24);
\path[fill=mycolor0] (29,23) rectangle (31,24);
\path[fill=mycolor6] (31,23) rectangle (33,24);
\path[fill=mycolor0] (33,23) rectangle (70,24);
\path[fill=mycolor0] (0,24) rectangle (8,25);
\path[fill=mycolor1] (8,24) rectangle (10,25);
\path[fill=mycolor0] (10,24) rectangle (70,25);
\path[fill=mycolor0] (0,25) rectangle (8,26);
\path[fill=mycolor1] (8,25) rectangle (10,26);
\path[fill=mycolor0] (10,25) rectangle (70,26);
\path[fill=mycolor0] (0,26) rectangle (8,27);
\path[fill=mycolor1] (8,26) rectangle (10,27);
\path[fill=mycolor0] (10,26) rectangle (47,27);
\path[fill=mycolor1] (47,26) rectangle (50,27);
\path[fill=mycolor0] (50,26) rectangle (70,27);
\path[fill=mycolor0] (0,27) rectangle (8,28);
\path[fill=mycolor1] (8,27) rectangle (10,28);
\path[fill=mycolor0] (10,27) rectangle (47,28);
\path[fill=mycolor1] (47,27) rectangle (50,28);
\path[fill=mycolor0] (50,27) rectangle (70,28);
\path[fill=mycolor0] (0,28) rectangle (8,29);
\path[fill=mycolor1] (8,28) rectangle (10,29);
\path[fill=mycolor0] (10,28) rectangle (47,29);
\path[fill=mycolor1] (47,28) rectangle (50,29);
\path[fill=mycolor0] (50,28) rectangle (70,29);
\path[fill=mycolor0] (0,29) rectangle (8,30);
\path[fill=mycolor1] (8,29) rectangle (10,30);
\path[fill=mycolor0] (10,29) rectangle (27,30);
\path[fill=mycolor8] (27,29) rectangle (34,30);
\path[fill=mycolor0] (34,29) rectangle (47,30);
\path[fill=mycolor1] (47,29) rectangle (50,30);
\path[fill=mycolor0] (50,29) rectangle (70,30);
\path[fill=mycolor0] (0,30) rectangle (8,31);
\path[fill=mycolor1] (8,30) rectangle (10,31);
\path[fill=mycolor0] (10,30) rectangle (27,31);
\path[fill=mycolor8] (27,30) rectangle (34,31);
\path[fill=mycolor0] (34,30) rectangle (47,31);
\path[fill=mycolor1] (47,30) rectangle (50,31);
\path[fill=mycolor0] (50,30) rectangle (70,31);
\path[fill=mycolor0] (0,31) rectangle (8,32);
\path[fill=mycolor1] (8,31) rectangle (10,32);
\path[fill=mycolor0] (10,31) rectangle (27,32);
\path[fill=mycolor8] (27,31) rectangle (34,32);
\path[fill=mycolor0] (34,31) rectangle (47,32);
\path[fill=mycolor1] (47,31) rectangle (50,32);
\path[fill=mycolor0] (50,31) rectangle (70,32);
\path[fill=mycolor0] (0,32) rectangle (8,33);
\path[fill=mycolor1] (8,32) rectangle (10,33);
\path[fill=mycolor0] (10,32) rectangle (27,33);
\path[fill=mycolor8] (27,32) rectangle (34,33);
\path[fill=mycolor0] (34,32) rectangle (47,33);
\path[fill=mycolor1] (47,32) rectangle (50,33);
\path[fill=mycolor0] (50,32) rectangle (70,33);
\path[fill=mycolor0] (0,33) rectangle (8,34);
\path[fill=mycolor1] (8,33) rectangle (10,34);
\path[fill=mycolor0] (10,33) rectangle (27,34);
\path[fill=mycolor8] (27,33) rectangle (34,34);
\path[fill=mycolor0] (34,33) rectangle (70,34);
\path[fill=mycolor0] (0,34) rectangle (8,35);
\path[fill=mycolor1] (8,34) rectangle (10,35);
\path[fill=mycolor0] (10,34) rectangle (27,35);
\path[fill=mycolor8] (27,34) rectangle (34,35);
\path[fill=mycolor0] (34,34) rectangle (70,35);
\path[fill=mycolor0] (0,35) rectangle (27,36);
\path[fill=mycolor8] (27,35) rectangle (34,36);
\path[fill=mycolor0] (34,35) rectangle (70,36);
\path[fill=mycolor0] (0,36) rectangle (27,37);
\path[fill=mycolor8] (27,36) rectangle (34,37);
\path[fill=mycolor0] (34,36) rectangle (70,37);
\path[fill=mycolor0] (0,37) rectangle (27,38);
\path[fill=mycolor8] (27,37) rectangle (34,38);
\path[fill=mycolor0] (34,37) rectangle (70,38);
\path[fill=mycolor0] (0,38) rectangle (27,39);
\path[fill=mycolor8] (27,38) rectangle (34,39);
\path[fill=mycolor0] (34,38) rectangle (70,39);
\path[fill=mycolor0] (0,39) rectangle (27,40);
\path[fill=mycolor8] (27,39) rectangle (34,40);
\path[fill=mycolor0] (34,39) rectangle (70,40);
\draw ( 0,0) -- ( 0 , 40 );
\draw ( 1,0) -- ( 1 , 40 );
\draw ( 2,0) -- ( 2 , 40 );
\draw ( 3,0) -- ( 3 , 40 );
\draw ( 4,0) -- ( 4 , 40 );
\draw ( 5,0) -- ( 5 , 40 );
\draw ( 6,0) -- ( 6 , 40 );
\draw ( 7,0) -- ( 7 , 40 );
\draw ( 8,0) -- ( 8 , 40 );
\draw ( 9,0) -- ( 9 , 40 );
\draw ( 10,0) -- ( 10 , 40 );
\draw ( 11,0) -- ( 11 , 40 );
\draw ( 12,0) -- ( 12 , 40 );
\draw ( 13,0) -- ( 13 , 40 );
\draw ( 14,0) -- ( 14 , 40 );
\draw ( 15,0) -- ( 15 , 40 );
\draw ( 16,0) -- ( 16 , 40 );
\draw ( 17,0) -- ( 17 , 40 );
\draw ( 18,0) -- ( 18 , 40 );
\draw ( 19,0) -- ( 19 , 40 );
\draw ( 20,0) -- ( 20 , 40 );
\draw ( 21,0) -- ( 21 , 40 );
\draw ( 22,0) -- ( 22 , 40 );
\draw ( 23,0) -- ( 23 , 40 );
\draw ( 24,0) -- ( 24 , 40 );
\draw ( 25,0) -- ( 25 , 40 );
\draw ( 26,0) -- ( 26 , 40 );
\draw ( 27,0) -- ( 27 , 40 );
\draw ( 28,0) -- ( 28 , 40 );
\draw ( 29,0) -- ( 29 , 40 );
\draw ( 30,0) -- ( 30 , 40 );
\draw ( 31,0) -- ( 31 , 40 );
\draw ( 32,0) -- ( 32 , 40 );
\draw ( 33,0) -- ( 33 , 40 );
\draw ( 34,0) -- ( 34 , 40 );
\draw ( 35,0) -- ( 35 , 40 );
\draw ( 36,0) -- ( 36 , 40 );
\draw ( 37,0) -- ( 37 , 40 );
\draw ( 38,0) -- ( 38 , 40 );
\draw ( 39,0) -- ( 39 , 40 );
\draw ( 40,0) -- ( 40 , 40 );
\draw ( 41,0) -- ( 41 , 40 );
\draw ( 42,0) -- ( 42 , 40 );
\draw ( 43,0) -- ( 43 , 40 );
\draw ( 44,0) -- ( 44 , 40 );
\draw ( 45,0) -- ( 45 , 40 );
\draw ( 46,0) -- ( 46 , 40 );
\draw ( 47,0) -- ( 47 , 40 );
\draw ( 48,0) -- ( 48 , 40 );
\draw ( 49,0) -- ( 49 , 40 );
\draw ( 50,0) -- ( 50 , 40 );
\draw ( 51,0) -- ( 51 , 40 );
\draw ( 52,0) -- ( 52 , 40 );
\draw ( 53,0) -- ( 53 , 40 );
\draw ( 54,0) -- ( 54 , 40 );
\draw ( 55,0) -- ( 55 , 40 );
\draw ( 56,0) -- ( 56 , 40 );
\draw ( 57,0) -- ( 57 , 40 );
\draw ( 58,0) -- ( 58 , 40 );
\draw ( 59,0) -- ( 59 , 40 );
\draw ( 60,0) -- ( 60 , 40 );
\draw ( 61,0) -- ( 61 , 40 );
\draw ( 62,0) -- ( 62 , 40 );
\draw ( 63,0) -- ( 63 , 40 );
\draw ( 64,0) -- ( 64 , 40 );
\draw ( 65,0) -- ( 65 , 40 );
\draw ( 66,0) -- ( 66 , 40 );
\draw ( 67,0) -- ( 67 , 40 );
\draw ( 68,0) -- ( 68 , 40 );
\draw ( 69,0) -- ( 69 , 40 );
\draw ( 70,0) -- ( 70 , 40 );
\draw (0, 0) -- ( 70 , 0 );
\draw (0, 1) -- ( 70 , 1 );
\draw (0, 2) -- ( 70 , 2 );
\draw (0, 3) -- ( 70 , 3 );
\draw (0, 4) -- ( 70 , 4 );
\draw (0, 5) -- ( 70 , 5 );
\draw (0, 6) -- ( 70 , 6 );
\draw (0, 7) -- ( 70 , 7 );
\draw (0, 8) -- ( 70 , 8 );
\draw (0, 9) -- ( 70 , 9 );
\draw (0, 10) -- ( 70 , 10 );
\draw (0, 11) -- ( 70 , 11 );
\draw (0, 12) -- ( 70 , 12 );
\draw (0, 13) -- ( 70 , 13 );
\draw (0, 14) -- ( 70 , 14 );
\draw (0, 15) -- ( 70 , 15 );
\draw (0, 16) -- ( 70 , 16 );
\draw (0, 17) -- ( 70 , 17 );
\draw (0, 18) -- ( 70 , 18 );
\draw (0, 19) -- ( 70 , 19 );
\draw (0, 20) -- ( 70 , 20 );
\draw (0, 21) -- ( 70 , 21 );
\draw (0, 22) -- ( 70 , 22 );
\draw (0, 23) -- ( 70 , 23 );
\draw (0, 24) -- ( 70 , 24 );
\draw (0, 25) -- ( 70 , 25 );
\draw (0, 26) -- ( 70 , 26 );
\draw (0, 27) -- ( 70 , 27 );
\draw (0, 28) -- ( 70 , 28 );
\draw (0, 29) -- ( 70 , 29 );
\draw (0, 30) -- ( 70 , 30 );
\draw (0, 31) -- ( 70 , 31 );
\draw (0, 32) -- ( 70 , 32 );
\draw (0, 33) -- ( 70 , 33 );
\draw (0, 34) -- ( 70 , 34 );
\draw (0, 35) -- ( 70 , 35 );
\draw (0, 36) -- ( 70 , 36 );
\draw (0, 37) -- ( 70 , 37 );
\draw (0, 38) -- ( 70 , 38 );
\draw (0, 39) -- ( 70 , 39 );
\draw (0, 40) -- ( 70 , 40 );
\draw[fill=red!80!yellow,very thick] (30.5,25.5) circle (0.35cm);
\end{tikzpicture}
}
\caption{Workspace for the single-robot example.}
\label{fig:simpleRobot}
\end{figure}

The specification for the robot is represented as a 15 state parity automaton. It encodes four conditions to hold:
\begin{itemize}
\item The left-most marked part of the workspace should be visited infinitely often,
\item the right-most marked part of the workspace should be visited infinitely often,
\item either the top marked part of the workspace must be visited only finitely often, or the bottom one, or both, and 
\item infinitely often, the regions in the middle shall be visited strictly in the middle-left-right order. 
\end{itemize}
The product MDP of the MDP and the parity automaton has 366015 states, out of which
2196 are unreachable (and can be removed).
\end{example}

{
A classical problem over MDPs with $\omega$-regular optimization criteria is to find a policy that maximizes the probability that a trace is accepting. In the product MDP from Example~\ref{example:motivation}, there is however no policy that raises this probability above $0$. This follows from the fact that no matter what the policy does, with a probability of at least $0.2$, the robot continues to travel into the current direction. By the limited size of the workspace, colliding with the workspace boundaries takes at most 35 steps, and thus, a very conservative lower bound on the probability for a crash within 35 steps is $(0.2)^{35}$ at \emph{every} step of the MDPs execution. In the long run, the collision is thus unavoidable with probability $1$.

Despite the fact that the parity MDP does not admit a good policy in the traditional sense, we may want to compute a policy that works towards the satisfaction of the specification as long as possible while avoiding unnecessary \emph{risks}. We formalize this objective in the following definition:
}

\begin{definition}
\label{def:mainProblem}
Let $\mathcal{M} = (S,A,\Sigma,P,C,s_0)$ be a parity MDP. We say that some control policy $f : S^* \rightarrow A$ has a \newterm{risk-averseness probability} $p \in [0,1]$ if there exist labelings $l : S^* \rightarrow \mathbb{N}$ and $l': S^* \rightarrow \mathbb{B}$ and a Markov chain $\mathcal{C}'$ induced by $\mathcal{M}$ and $f$ with the following properties:
\begin{itemize}
\item There exists some number $k \in \NN$ such that for all $t_0 t_1 t_2 \ldots \in S^\omega$, there are at most $k$ many indices $i \in \NN$ for which we have $l(t_0 \ldots t_i) > l(t_0 \ldots t_i t_{i+1})$.
\item For all $t_0 t_1 \ldots t_n \in S^*$, we have that $l(t_0 \ldots t_n)$ is even, and $l'(t_0 \ldots t_n)=\TRUE$ implies that $C(t_n) \geq l(t_0 \ldots t_n)$ and that $C(t_n)$ is even.
\item For all $t_0 t_1 \ldots t_n \in S^*$, if $C(t_n)$ is odd, then $l(t_0 \ldots t_n) > C(t_n)$.
\item For all  $t = t_0 t_1 \ldots t_n \in S^*$ with either (a) $l'(t) = \TRUE$ or (b) $t = s_0$, the probability measure in $\mathcal{C}'$ to reach some state $t\,t'_0 \ldots t'_m \in S^*$ with $l'(t\,t'_0 \ldots t'_m)=\TRUE$ from state $t$ is at least $p$.
\end{itemize}
\end{definition}
The labellings $l$ and $l'$ in Definition~\ref{def:mainProblem} augment a policy with the information what \newterm{goal color} the policy is trying to reach and when a goal has been reached. A goal must always be even-colored, but along different traces, different goals are allowed. 
From every goal state, the next goal state must be reached with probability at least $p$. Together with the first two requirements in Definition~\ref{def:mainProblem}, this implements the parity acceptance condition, as they together state that the {goal color} can only decrease finitely often along a trace. 
The parity acceptance condition does not need to be fulfilled with strictly positive probability in the long run, however, as in between two visits to goal states, the policy may fail with probability $(1-p)$. Thus, we only require the parity acceptance condition to hold on those paths on which goal states are reached infinitely often (which may have probability measure $0$).
{The strategy can choose goal states in a way that maximizes the probability of reaching the respective next goal state. Thus, the higher the value of $p$ is, the more averse to the risk to miss the next goal the control policy needs to be.}

The reader may wonder why mentioning the labeling function $l'$ is actually necessary in Definition~\ref{def:mainProblem}, as one could simply implicitly set $l'(t_0 \ldots t_n) = \TRUE$ whenever $C(t_n) \geq l(t_0 \ldots t_n)$ and $C(t_n)$ is even. However, this change requires the policy to be able to reach the next goal from state $t_n$ with probability $p$ in the induced Markov chain, which is not always possible in a $p$-risk-averse strategy. Figure~\ref{fig:nonWinningExample} shows an example in which increasing the color of state $q_1$ to $2$ (which is even) would reduce the maximally implementable risk-averseness level from $0.68$ to $0.64$. As changing an odd color to an even one only makes the parity acceptance condition easier to satisfy, this is a very unintuitive property. To avoid it, we thus chose to make the labeling function $l'$ explicit. 

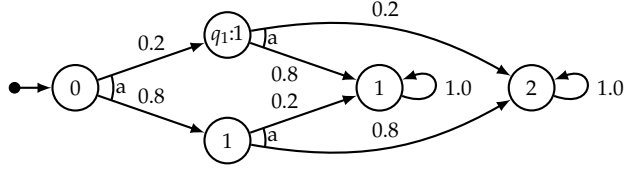
\begin{figure}
\centering
\begin{tikzpicture}[auto,node distance=8mm,>=latex,font=\scriptsize,myrotate/.style 2 args={rotate=#1,nodes={rotate=#2}}]
\tikzstyle{round}=[thick,draw=black,circle]
\tikzset{every loop/.style={min distance=10mm,in=0,out=60,looseness=10}}

\path[use as bounding box] (-0.87,-0.95) rectangle (7.13,1.15);

\draw[fill=black] (-0.8,0.0) circle (0.07cm);
\node[round] (s0) at (0,0) {$0$};
\draw[thick,->] (-0.8,0.0) -- (s0);
\node[round] (s1) at (2,0.7) {{\color{white}$1$}};
\node at (2,0.7) {$q_1\!\!:\!\!1$};
\node[round] (s2) at (2,-0.7) {$1$};
\node[round] (s3) at (4,0) {$1$};
\node[round] (s4) at (6,0) {$2$};

\draw[thick,->] (s0) edge node[above] {0.2} coordinate[pos=0.1](sA1C1) (s1);
\draw[thick,->] (s0) edge node[above] {0.8} coordinate[pos=0.1](sA1C2) (s2);
\draw[thick,fill=red,draw] pic[draw,angle radius=5mm,"a",angle eccentricity=1.2] {angle = sA1C2--s0--sA1C1};

\draw[thick,->] (s1) edge node[below left=-1mm] {0.8} coordinate[pos=0.1](sA1C2) (s3);
\draw[thick,->] (s1) edge[bend left=20] node[above] {0.2} coordinate[pos=0.1](sA1C1) (s4);
\draw[thick,fill=red,draw] pic[draw,angle radius=5mm,"a",angle eccentricity=1.2] {angle = sA1C2--s1--sA1C1};

\draw[thick,->] (s2) edge node[above left=-1mm] {0.2} coordinate[pos=0.1](sA1C2) (s3);
\draw[thick,->] (s2) edge[bend right=20] node[above] {0.8} coordinate[pos=0.1](sA1C1) (s4);
\draw[thick,fill=red,draw] pic[draw,angle radius=5mm,"a",angle eccentricity=1.2] {angle = sA1C1--s2--sA1C2};

\draw[thick,->] (s3) edge[loop,out=-20,in=20,distance=0.6cm] node[right] {1.0} (s3);
\draw[thick,->] (s4) edge[loop,out=-20,in=20,distance=0.6cm] node[right] {1.0} (s4);

\end{tikzpicture}
\caption{An MDP in which for risk-averseness level $p=0.68$, state $q_1$ is not winning, but the state is reachable on the (unique) $p$-risk-averse policy. All states are labeled by their colors. }
\label{fig:nonWinningExample}
\end{figure}
{Using Definition~\ref{def:mainProblem}, we can now state the main problem considered in this paper:}
\begin{definition}[Optimal risk-averse policy synthesis]
Given a parity MDP, the optimal risk-averse policy synthesis is to find the highest value $p$ such that a policy for the MDP with risk-averseness level $p$ exists, and to find such a policy.
\label{def:optimalRiskAversePolicyComputation}
\end{definition}

\section{Computing Risk-Averse Policies}
In this section, we describe an algorithm to compute risk-averse policies in parity MDPs. The algorithm produces finite-memory strategies that are not necessarily positional. This may appear to be a flaw of the algorithm, as memoryless policies suffice for maximizing the probability for a trace to satisfy a parity objective in MDPs \cite{chatterjee2004quantitative}. 
However, optimal risk-averse strategies \emph{do} require memory in general, which we show by means of an example.

\begin{figure}
\centering
\begin{tikzpicture}[auto,node distance=8mm,>=latex,font=\scriptsize,myrotate/.style 2 args={rotate=#1,nodes={rotate=#2}}]
\tikzstyle{round}=[thick,draw=black,circle]
\tikzset{every loop/.style={min distance=10mm,in=0,out=60,looseness=10}}

\path[use as bounding box] (-4.1,-4.25) rectangle (4.5,4.05);

\node[round] (s0) at (0,0) {$3$};

\begin{scope}[myrotate={-18}{-18}]
\node[round] (sA1) at (2,0) {$0$};
\node[round] (sB1) at (4,0) {$1$};
\draw[thick,->] (s0) -- node[above] {0.9} (sA1);
\draw[thick,->] (s0) edge[bend right=20] node[below] {0.1} coordinate[pos=0.1](sB1C) (sB1);
\draw[thick,fill=red,draw] pic[draw,angle radius=10mm,"a",angle eccentricity=1.2] {angle = sB1C--s0--sA1};
\draw[thick,->] (sB1) edge[loop,out=-20,in=20,distance=0.6cm] node[below=1mm,xshift=-3pt] {1.0} (sB1);
\draw[thick,->] (sA1) edge[bend right=30] node[above] {0.7} coordinate[pos=0.1](sA1C1) (s0);
\draw[thick,->] (sA1) edge[bend left=30] node[above] {0.3} coordinate[pos=0.1](sA1C2) (sB1);
\draw[thick,fill=red,draw] pic[draw,angle radius=5mm,"a",angle eccentricity=1.2] {angle = sA1C2--sA1--sA1C1};
\end{scope}

\begin{scope}[myrotate={56}{56}]
\node[round] (sA1) at (2,0) {$0$};
\node[round] (sB1) at (4,0) {$1$};
\draw[thick,->] (s0) -- node[above] {0.8} (sA1);
\draw[thick,->] (s0) edge[bend right=20] node[below] {0.2} coordinate[pos=0.1](sB1C) (sB1);
\draw[thick,fill=red,draw] pic[draw,angle radius=10mm,"b",angle eccentricity=1.2] {angle = sB1C--s0--sA1};
\draw[thick,->] (sB1) edge[loop,out=-20,in=20,distance=0.6cm] node[below=2mm,xshift=-2pt] {1.0} (sB1);
\draw[thick,->] (sA1) edge[bend right=30] node[above] {0.8} coordinate[pos=0.1](sA1C1) (s0);
\draw[thick,->] (sA1) edge[bend left=30] node[above] {0.2} coordinate[pos=0.1](sA1C2) (sB1);
\draw[thick,fill=red,draw] pic[draw,angle radius=5mm,"a",angle eccentricity=1.2] {angle = sA1C2--sA1--sA1C1};
\end{scope}

\begin{scope}[myrotate={126}{126}]
\node[round] (sA1) at (2,0) {\rotatebox{180}{$0$}};
\node[round] (sB1) at (4,0) {\rotatebox{180}{$1$}};
\draw[thick,->] (s0) -- node[above] {\rotatebox{180}{0.7}} (sA1);
\draw[thick,->] (s0) edge[bend right=20] node[below] {\rotatebox{180}{0.3}} coordinate[pos=0.1](sB1C) (sB1);
\draw[thick,fill=red,draw] pic[draw,angle radius=10mm,"\rotatebox{180}{c}",angle eccentricity=1.2] {angle = sB1C--s0--sA1};
\draw[thick,->] (sB1) edge[loop,out=-20,in=20,distance=0.6cm] node[left=5mm,xshift=-2pt] {\rotatebox{180}{1.0}} (sB1);
\draw[thick,->] (sA1) edge[bend right=30] node[above] {\rotatebox{180}{0.9}} coordinate[pos=0.1](sA1C1) (s0);
\draw[thick,->] (sA1) edge[bend left=30] node[above] {\rotatebox{180}{0.1}} coordinate[pos=0.1](sA1C2) (sB1);
\draw[thick,fill=red,draw] pic[draw,angle radius=5mm,"\rotatebox{180}{a}",angle eccentricity=1.2] {angle = sA1C2--sA1--sA1C1};
\end{scope}

\begin{scope}[myrotate={198}{198}]
\draw[fill=black] (4,0.8) circle (0.07cm);
\node[round] (sA1) at (4,0) {\rotatebox{180}{$0$}};
\draw[thick,->] (4,0.8) -- (sA1);
\node[round] (sB1) at (2,0) {\rotatebox{180}{$1$}};
\draw[thick,->] (sB1) edge[loop,out=160,in=200,distance=0.6cm] node[left] {\rotatebox{180}{$1.0$}} (sB1);
\draw[thick,->] (sA1) edge[bend left=30] node[below] {\rotatebox{180}{0.4}} coordinate[pos=0.1](sA1C1) (sB1);
\draw[thick,->] (sA1) edge[bend right=20] node[above] {\rotatebox{180}{0.6}} coordinate[pos=0.1](sA1C2) (s0);
\draw[thick,fill=red,draw] pic[draw,angle radius=5mm,"\rotatebox{180}{a}",angle eccentricity=1.2] {angle = sA1C2--sA1--sA1C1};
\end{scope}

\begin{scope}[myrotate={270}{0}]
\node[round] (sA1) at (2,0) {$1$};
\node[round] (sB1) at (4,0) {$2$};
\draw[thick,->] (s0) edge[bend left=20] node[right] {$0.6$} coordinate[pos=0.1](sA1D1) (sB1);
\draw[thick,->] (s0) edge node[left] {0.4} coordinate[pos=0.1](sA1D2) (sA1);

\draw[thick,fill=red,draw] pic[draw,angle radius=10mm,"d",angle eccentricity=1.2] {angle = sA1D2--s0--sA1D1};
\end{scope}
\draw[thick,->] (sA1) edge[loop,out=165,in=195,distance=0.6cm] node[left] {$1.0$} (sA1);
\draw[thick,->] (sB1) edge[loop,out=-15,in=15,distance=0.6cm] node[right] {$1.0$} (sB1);

\end{tikzpicture}
\caption{An example parity MDP that admits a $0.54$-risk-averse finite-memory policy, but no such positional policy. All states are labeled by their colors.}
\label{fig:exampleMDP}
\end{figure}

\begin{example}
Figure~\ref{fig:exampleMDP} shows a parity MDP. It has four colors, and all states with color $1$ are sink states, i.e., from which no possible goal state can be reached.
The center state has the highest and odd color, so it may only be visited finitely often. 
Any policy cannot avoid either ending up in a sink state or visiting the middle state at least every second step, unless eventually action $d$ is chosen by the policy. 
If the policy chooses action $a$ in the initial state, and then immediately chooses $d$, it reaches the state with color $2$ with a probability of $0.6 \cdot 0.6 = 0.36$. The resulting policy is thus $0.36$-risk averse. However, there exists a better policy: when the state with color $3$ is visited for the first time, action $a$ should be taken, then action $b$, $c$, and finally action $d$. By declaring all color $0$ states to be goal states, the resulting policy then has a risk averseness level of $\min(0.6 \cdot 0.9, 0.7 \cdot 0.8, 0.8 \cdot 0.7, 0.9 \cdot 0.6) = 0.54$. Thus, the best next action in the state with color $3$ depends on the history of the trace.
While this example only shows that memory is needed in optimally risk-averse policies, the fact that finite memory suffices follows from the correctness of our algorithm described below.
\end{example}

\subsection{$p$-risk-averse policy computation}
\label{subsec:pRiskAversePolicyComputation}

Let us assume that $p$ is fixed and that we want to compute a $p$-risk-averse MDP control policy. The algorithm that we describe in this section computes the set of states from which a $p$-risk-averse policy exists. We call such states \newterm{winning}. The policies computed sometimes make use of non-winning states, which may be counter-intuitive at first. 
Figure~\ref{fig:nonWinningExample} shows an example MDP where this is the case: from state $q_1$, the probability of reaching a next goal state is only $0.2$, but the optimal $0.68$-risk-averse policy from the initial state requires that even after reaching $q_1$, state $q_2$ is labeled as being a goal state if it is subsequently reached.

Whenever a goal state is reached, the only information about the history of the trace that may need to be retained is (1) how often the goal color may still be decreased before the limit of $k$ is reached, and (2) what the current goal color is. This follows from the fact that the computation of probabilities is \emph{reset} at goal states. Our algorithm makes use of this fact by planning policies from goal state to goal state(s).
It iterates over all possible value combinations for the current goal color and the number of remaining goal color reductions. 

\begin{definition}
We say that a state $q$ is $(k,c)$-winning (for some fixed risk averseness level $p$) if there exists a $p$-risk-averse policy $f$ from $q$ as initial state with labels $l$ and $l'$ such that $l(\epsilon)=c$ and along all traces of the policy, goal colors are never decreased more than $k$ times. We call such policies $p$-$(k,c)$-risk-averse.
\label{def:ckWinning}
\end{definition}

\begin{corollary}
Some parity MDP admits a $p$-risk-averse policy if the initial state is $(k,c)$-winning for some values of $c \in \NN$ and $k \in \NN$.
\end{corollary}
\begin{proof}
Follows directly from Definitions~\ref{def:mainProblem} and \ref{def:ckWinning}.
\end{proof}

This corollary allows us to frame the search for a $p$-risk-averse policy as an iterative process, which we base on the following lemma. Let in the following $c_\mathit{maxEven}$ be the least \emph{even} upper bound on the colors occurring in the parity MDP.

\begin{lemma}
A state $q$ is winning for some values of $(k,c)$ with $c \leq c_\mathit{maxEven}$ and even $c$ if and only if there exists a policy such that with probability $p$ eventually either:
\begin{itemize}
\item some even-colored state $q'$ is visited that is winning for $(k-1,0)$, or
\item some even-colored state $q'$ with $C(q') = c'$ for $c' \geq c$ is visited that is winning for $(k,c')$ while no odd color $\geq c'$ is visited along the way to $q'$.
\end{itemize}
\label{lem:ckLemma}
\end{lemma}
\begin{proof}
We prove the claim by induction over $(k,c)$ with even $c$. The order of induction that we use is lexicographic in $(k,-1 \cdot c)$.

\textbf{Induction basis:} For the case $(k,c) = (0,c_\mathit{maxEven})$, the only way for a state to be $(k,c)$-winning is for a policy from that state to exist such that with probability at least $p$, a state is eventually visited that has color $c$ and is $(k,c)$-winning again. This is exactly the only condition from the claim that is applicable in this case.

\textbf{Induction step:}
($\Rightarrow$)
Let $f'$ be a policy from $q$ such that on every trace of the policy, the  goal colors decrease at most $k$ times, and let $l(\epsilon)=c$ for the labellings $(l,l')$ assigned to the policy. 
The probability to reach the next goal must be at least $p$ in order for the state to be $(k,c)$-winning. A goal can either have a color of $\geq c$ or a color less than $c$. In the latter case, the goal state must be $(k',c')$-winning for some $c'$ and some $k'<k$. As all such states are also $(k-1,0)$-winning (by definition), this case is covered by cases in the claim.
If a goal with color $c'$ is reached, then either no state with color $\geq c$ is visited along the way and $c'=c$, or alternatively $c'>c$ and no state with color $\geq c'$ is visited along the way. Both cases are covered by the case list in the claim.

($\Leftarrow$)
{Now let $q$ be a state from which a policy to visit some goal state $q'$ with probability $p$ exists.}
State $q'$ can be a $(k,c)$-winning state, but does not need to be one. If on the way to $q'$, a state with an odd color $c' > c$ is visited, this requires that the label function $l'$ of the policy has to be greater than $c'$ on the way from $q$ to $q'$. So for the trace to count towards the probability mass of $p$, state $q'$ needs to be either $(k,c'+x)$-winning (for even $x\geq 2$) or alternatively $(k-1,c')$-winning. Since the set of $(k,c'+x)$-winning states is contained in the $(k,c'+2)$-winning states and the $(k-1,c')$-winning states are a subset of the $(k-1,0)$-winning states (by definition), we can assume, without loss of generality, that a $(k,c'+2)$-winning or $(k-1,0)$-winning state is visited.

We construct the $p$-risk averse policy $f$ with associated labels $(l,l')$ that prove that $q$ is $(k,c)$-winning as follows: we use the policy with the properties from the claim, and switch to the policies that exist by the inductive hypothesis for the states that are $(k-1,0)$-winning or $(k,c'+2)$-winning when the second condition from the claim is used. 
When another $(k,c)$-winning state is visited, we instead continue with a policy constructed from $q'$ in the same way as for $q$. 
The fact that this composition of the policies yields a correct $(k,c)$-winning policy follows by induction: at every policy prefix $t$ with $l'(t)=\TRUE$ or $t = \epsilon$ such that no transition to a $(k-1,0)$-winning or $(k,c'+2)$-winning has yet occurred, we know that the policy reaches some next goal state with probability at least $p$. For the other goal states, the correctness follows from the inductive hypothesis and the fact that after transitions to $(k,c'+2)$-winning or $(k-1,0)$-winning goal states, the existing $p$-risk-averse policies can be used from there. If no such other goal state is reached or until such a goal state is reached, the construction ensures that the goal states otherwise reached are $(k,c)$-winning, and no odd color higher than $c$ is reached in between two visits to $(k,c)$-winning goal states that are not $(k,c'+2)$-winning or $(k-1,0)$-winning. As this property holds (by induction over the length of the policy prefix) for all visits to goal states, the claim follows.
\end{proof}

The characterization of $(k,c)$-winning states in Lemma~\ref{lem:ckLemma} allows us to compute the $(k,c)$-winning states using traditional MDP policy computation algorithms.

\begin{lemma}
Let $\mathcal{M} = (S,A,\Sigma,P,C,s_0)$ be a parity MDP, $S_{k,c} \subseteq S$ be the $(k,c)$-winning states, $S_{k,c+2}, \allowbreak \ldots, \allowbreak S_{k,c_\mathit{maxEven}}$ be the $(k,c+2)$-winning to $(k,c_\mathit{maxEven})$-winnings states, and $S_{k-1,0}$ be the states that are $(k-1,0)$-winning (for some value of $p$). We can compute a reachability MDP $\mathcal{M'}$ with $|S| \times |\{c, c+2, \ldots, c_\mathit{maxEven}\}|$ many states in which the value of any state $(q,c)$ is $\geq p$ if and only if $q$ is a $(k,c)$-winning state.
\label{lem:reachabilityMDPLemma}
\end{lemma}
\begin{proof}
We can construct $\mathcal{M}' = (S',A,\Sigma,P',g,s_0)$ as follows:
{
\begin{align*}
S' & = S \times \{c, c+2, \ldots, c_\mathit{maxEven}\} \\
P((s,\tilde c),a)((s',\tilde c')) & = \begin{cases}
P(s,a)(s') & \text{if } C(s') \text{ is odd and } 
\tilde c' = \max(\tilde c, C(s')) \\
P(s,a)(s') & \text{if } C(s') \text{ is even and } 
\tilde c' = \tilde c \\
0 & \text{else}
\end{cases} \\
& \quad \quad \text{for all} (s,\tilde c), (s',\tilde c') \in S', a \in A
\end{align*}\begin{align*}
g((s,\tilde c)) & = \begin{cases} 1 & \text{if } \tilde c = c, s \in S_{k,c}, C(s)\geq \tilde c, C(s)\text{ is even} \\
1 & \text{if } s \in S_{k,\tilde c} \text{ or } s \in S_{k-1,0},\text{ and }C(s)\text{ is even}\\
0 & \text{else}
\end{cases} \\
& \quad \quad \text{for all} (s,\tilde c) \in S'
\end{align*}
} %
The MDP has the stated properties by the facts that (1) it keeps track of the highest color visited along a trace so far, and (2) it induces a payoff of $1$ exactly for the states that are possible goal states.
\end{proof}

Optimal policy computation for a reachability MDP can be performed by standard policy iteration or value iteration algorithms. Until now, the definition of the reachability MDP in Lemma~\ref{lem:reachabilityMDPLemma} is somewhat recursive: in order to determine which states are $(k,c)$-winning, we have to already know the $(k,c)$-winning states. The characterization from Lemma~\ref{lem:ckLemma} however allows us to compute it with the approach from Lemma~\ref{lem:reachabilityMDPLemma}. What we are actually searching for is the \emph{largest} set of states $S_\mathit{k,c}$ that the construction from Lemma~\ref{lem:reachabilityMDPLemma} maps to itself; any state set that is smaller misses some states that are $(k,c)$-winning by the characterization from Lemma~\ref{lem:ckLemma}, and by the same lemma, any set that is larger contains some state that is not $(k,c)$-winning. So computing the \emph{greatest fixpoint} over the states $Q_\mathit{k,c}$ allows to find the $(k,c)$-winning states, provided that the $(k,c+2)$-winning to $(k,c_\mathit{maxEven})$-winning and $(k-1,0)$-winning states are known. By iterating over the possible values of $k$ and $c$, we can thus compute the sets $S_{k,c}$ in a bottom-up fashion, as shown in Algorithm~\ref{algo:complexAlgorithm}. 

\begin{algorithm}
\begin{algorithmic}[1]
\Function{ComputeRAPolicy}{$\mathcal{M}, p$}
\State $S_{k-1} \gets \emptyset$
\While{fixed point of $S_k$ has not been reached}
\For{$c \in \{c_\mathit{maxEven},c_\mathit{maxEven}-2, \ldots, 0\}$}
\State $S_{k}[c] \gets S$
\While{fixed point of $S_k[c]$ has not been reached}
\State $\mathcal{M}' = \Call{ConstructionFromLemma2}{c, \allowbreak S_{k}[c], \allowbreak \ldots, \allowbreak S_{k}[c_\mathit{maxEven}],\allowbreak S_{k-1}}$
\State $V \gets \Call{ComputeStateValues}{\mathcal{M}'}$
\State $S_k[c] \gets \{s \in S \mid V((s,c))\geq p\}$ \label{algoline:computerSkc}
\EndWhile
\EndFor
\State $S_{k-1} \gets S_{k}[0]$
\EndWhile
\State \Return $s_0 \in S_{k-1}$
\EndFunction
\end{algorithmic}
\caption{Algorithm to compute if a parity MDP $\mathcal{M}$ admits a $p$-risk-averse policy.}
\label{algo:complexAlgorithm}
\end{algorithm}

The algorithm calls the external function \textsc{ComputeStateValues} to solve the reachability MDPs obtained by the construction in Lemma~\ref{lem:reachabilityMDPLemma}, which can be a value or policy iteration algorithm. 
Extending Algorithm~\ref{algo:complexAlgorithm} to also compute a policy is simple: without loss of generality, optimal reachability MDP policies are positional, and we can stitch these policies together in the order in which they are found by the algorithm.
Since the algorithm performs only a finite number of iterations over $k$ and $c$, the resulting policy is finite-state.

\begin{remark}
\label{remark:algorithmSimplication}
To speed up Algorithm~\ref{algo:complexAlgorithm}, we can simplify the reachability MDP construction of Lemma~\ref{lem:reachabilityMDPLemma}: instead of keeping track of the maximum odd color seen along a trace so far (in excess of $c$), we can alternatively keep track of whether an odd color greater than $c$ has been seen so far, and only consider switching to a $(k-1,0)$-winning goal state in that case. While the number of loop iterations of the algorithm until all positions that admit a $p$-risk-averse policy has been found can be higher with this modification, the reachability MDPs are typically smaller (as they have a size of at most $2 \cdot | S |$ then), which speeds up the value or policy iteration process for solving them. 
\end{remark}

\subsection{Maximally risk-averse policy computation}
\label{subsec:binarySearch}

In the previous subsection, we gave an algorithm to obtain $p$-risk-averse policies for a given $p$ whenever they exist. In order to compute optimally risk-averse policies, we can apply a \emph{bisection search}, which is the continuous-domain version of binary search, to find the highest value $p$ such that a $p$-risk-averse policy exists.

Since $p$ is a continuous value, this process has no natural termination point, however. For all practical means, it makes sense to define a cut-off value for the search such that if the difference between known upper and lower bounds on the risk-averseness level of the optimal policy is below the cut-off, the search process terminates with the best policy found until then. Defining a cut-off point is also motivated by practical means: most MDP solving algorithms run with a bounded precision, which leads to rounding errors. This makes it difficult to solve the problem given in Definition~\ref{def:optimalRiskAversePolicyComputation} in the strict sense.

However, under the assumption that the probabilities computed by function \textsc{ComputeStateValues} are exact, Algorithm~\ref{algo:complexAlgorithm} can be modified in order to allow finding a maximally risk-averse policy. For this, line~\ref{algoline:computerSkc} of the algorithm needs to be replaced by $S_k[c] \gets \{s \in S \mid V((s,c)) > p\}$. The algorithm then checks if a $p'$-risk-averse policy for $p'>p$ exists. Furthermore, after every call to \textsc{ComputeStateValues}, we let the algorithm also compute $\mathit{lb} := \min \{V(s) \mid s \in S, V((s,c))> p\}$. The least of these $\mathit{lb}$ values represents a lower bound on the $p$-risk averseness of the policy actually computed. Let this value be named $\mathit{lb}_\mathit{min}$.

We can now perform an iterative search process for the optimally risk-averse policy as follows: starting with $p=0$, we search for a $p'$-risk-averse policy for $p'>p$ using the modified version of Algorithm~\ref{algo:complexAlgorithm}. If we find one, we update $p$ to $\mathit{lb}_\mathit{min}$ and continue with the search. Otherwise, the previously found policy is an optimally risk-averse policy.

To see why this process solves the problem, note that whenever $p$ is increased, at least one state is removed from $S_k[c]$ in some iteration of the outermost while loop. While the state may be added to $S_k[c]$ \emph{later} in the process, increasing the value of $p$ can only push states to be found later in the search process of Algorithm~\ref{algo:complexAlgorithm}. When delaying the addition of states to $S_k[c]$, at some point, there will be one execution of the outer while loop of Algorithm~\ref{algo:complexAlgorithm} in which no additional states are found. Since the algorithm will terminate without finding a policy in this case, by the correctness of the algorithm, we can terminate the search at that point, and the policy found last is optimally risk-averse.

\section{Experiments}

We implemented the $p$-risk-averse policy computation approach in a prototype tool written in \texttt{C++} that is called \texttt{ramps}. The tool uses the simplification from Remark~\ref{remark:algorithmSimplication} and employs value iteration to compute policies for the reachability MDPs analyzed in Algorithm~\ref{algo:complexAlgorithm}. Bisection search with a cut-off value of $0.01$ (i.e., 1 percent) is used to computed close-to-optimal risk-averse policies. We configured the value iteration processes to terminate when the sum of updates to the state values falls below $0.05$.
Value iteration is performed in a parallelized way using the \texttt{openmp} library. All computation times reported in the following were taken on an Intel i5-4200U computer with 1.60\,GHz clock rate and 4GB of RAM, utilizing 2 physical processor cores, each with two virtual hyper-threaded cores that are made use of for value iteration. The \texttt{ramps} tool is available under the GPLv3 open source license from \texttt{https://github.com/progirep/ramps}.

\subsection{Single-robot control}
{In the first experiment, we consider the setting from Example~\ref{example:motivation}. The \texttt{ramps} tool needs} 30 minutes and 11 seconds (95m57s of single-processor time) to compute a $0.890689$-risk-averse policy with 388329 states. A simulation of it, available as a video on \texttt{https://progirep.github.\allowbreak{}io/\allowbreak{}ramps}, shows that the robot performs the task encoded into the parity automaton until it crashes. Visiting the regions in the middle in the correct order seems to be relatively easy for the policy. In order to reach the regions on the left and on the right in a risk-averse way, the robot often circles many times before it has the right approach angle and position to travel through one of the gaps next to the static obstacles.

\subsection{Multi-robot control}
As a second example, we considered a multi-robot control scenario, which we depict in Figure~\ref{fig:exampleworkspace2}. This time, we have two robots without complex dynamics: in each step, they can either move left, right, up or down by one cell, or choose not to move at all. If a robot chooses to move, there is an 8 percent chance that it moves into a different direction than chosen (i.e., $8/3$ percent per remaining direction). As in the first example, crashing into an obstacle or into the workspace boundaries leads to a transition to an error state in the MDP.
A robot crashing into the other robot also leads to the error state.
 
The robots can also carry an item. For this, they have to jointly perform a pickup operation while standing left and right, respectively, of the pickup region $r_1$. While they maintain a horizontal distance of $2$, they can continue carrying the item. The item is lost if there is a deviation in the distance. At region $r_2$, they can also drop the item. They cannot crash into each other while carrying an item (as it acts like a spacer). The MDP has 12294 states, 307304 state/action pairs, and 2798040 edges. The numbers of state/action pairs and edges are higher than in the first scenario, as each of the two robots has five choices of actions in each step.

The specification is represented as a 5-state parity automaton that encodes that (1) infinitely many items shall be delivered from $r_1$ to $r_2$, (2) infinitely often, robot one and two shall visit the top left and top right regions, respectively, and (3) the pickup and dropping regions should never be visited by any robot. 

Computing a $0.599408$-risk-averse policy takes $146.4$ seconds and the simulation (available as a video on \texttt{https://progirep.github.\allowbreak{}io/\allowbreak{}ramps}) shows that again, the policy lets the robots perform their task until at least one of them collides. In case the item is lost during delivery, the two robots just try again immediately. The generated policy has $61509$ states.

\begin{figure}
\centering \resizebox{0.38\columnwidth}{!}{%
\begin{tikzpicture}[yscale=-1]
\definecolor{mycolor0}{RGB}{255,255,255}
\definecolor{mycolor1}{RGB}{0,0,0}
\definecolor{mycolor2}{RGB}{0,255,0}
\definecolor{mycolor3}{RGB}{255,0,1}
\definecolor{mycolor4}{RGB}{0,0,255}
\definecolor{mycolor5}{RGB}{255,0,255}
\definecolor{mycolor6}{RGB}{0,255,255}
\definecolor{mycolor7}{RGB}{220,220,220}
\definecolor{mycolor8}{RGB}{205,206,179}
\path[fill=mycolor4] (0,0) rectangle (2,1);
\path[fill=mycolor0] (2,0) rectangle (7,1);
\path[fill=mycolor1] (7,0) rectangle (8,1);
\path[fill=mycolor0] (8,0) rectangle (9,1);
\path[fill=mycolor5] (9,0) rectangle (11,1);
\path[fill=mycolor4] (0,1) rectangle (2,2);
\path[fill=mycolor0] (2,1) rectangle (7,2);
\path[fill=mycolor1] (7,1) rectangle (8,2);
\path[fill=mycolor0] (8,1) rectangle (9,2);
\path[fill=mycolor5] (9,1) rectangle (11,2);
\path[fill=mycolor0] (0,2) rectangle (3,3);
\path[fill=mycolor1] (3,2) rectangle (8,3);
\path[fill=mycolor0] (8,2) rectangle (11,3);
\path[fill=mycolor0] (0,3) rectangle (11,4);
\path[fill=mycolor0] (0,4) rectangle (11,5);
\path[fill=mycolor0] (0,5) rectangle (2,6);
\path[fill=mycolor1] (2,5) rectangle (3,6);
\path[fill=mycolor0] (3,5) rectangle (11,6);
\path[fill=mycolor0] (0,6) rectangle (2,7);
\path[fill=mycolor1] (2,6) rectangle (3,7);
\path[fill=mycolor0] (3,6) rectangle (11,7);
\path[fill=mycolor0] (0,7) rectangle (8,8);
\path[fill=mycolor3] (8,7) rectangle (9,8);
\path[fill=mycolor0] (9,7) rectangle (11,8);
\path[fill=mycolor0] (0,8) rectangle (2,9);
\path[fill=mycolor2] (2,8) rectangle (3,9);
\path[fill=mycolor0] (3,8) rectangle (11,9);
\path[fill=mycolor0] (0,9) rectangle (8,10);
\path[fill=mycolor1] (8,9) rectangle (9,10);
\path[fill=mycolor0] (9,9) rectangle (11,10);
\path[fill=mycolor0] (0,10) rectangle (11,11);

\node at (2.5,8.5) {\Large $r_1$};
\node at (8.5,7.5) {\Large $r_2$};

\draw ( 0,0) -- ( 0 , 11 );
\draw ( 1,0) -- ( 1 , 11 );
\draw ( 2,0) -- ( 2 , 11 );
\draw ( 3,0) -- ( 3 , 11 );
\draw ( 4,0) -- ( 4 , 11 );
\draw ( 5,0) -- ( 5 , 11 );
\draw ( 6,0) -- ( 6 , 11 );
\draw ( 7,0) -- ( 7 , 11 );
\draw ( 8,0) -- ( 8 , 11 );
\draw ( 9,0) -- ( 9 , 11 );
\draw ( 10,0) -- ( 10 , 11 );
\draw ( 11,0) -- ( 11 , 11 );
\draw (0, 0) -- ( 11 , 0 );
\draw (0, 1) -- ( 11 , 1 );
\draw (0, 2) -- ( 11 , 2 );
\draw (0, 3) -- ( 11 , 3 );
\draw (0, 4) -- ( 11 , 4 );
\draw (0, 5) -- ( 11 , 5 );
\draw (0, 6) -- ( 11 , 6 );
\draw (0, 7) -- ( 11 , 7 );
\draw (0, 8) -- ( 11 , 8 );
\draw (0, 9) -- ( 11 , 9 );
\draw (0, 10) -- ( 11 , 10 );
\draw (0, 11) -- ( 11 , 11 );
\draw[fill=red!80!yellow,very thick] (5.5,5.5) circle (0.35cm);
\draw[fill=red!80!yellow,very thick] (1.5,1.5) circle (0.35cm);
\end{tikzpicture}
}
\caption{Workspace for the multi-robot example.}
\label{fig:exampleworkspace2}
\end{figure}
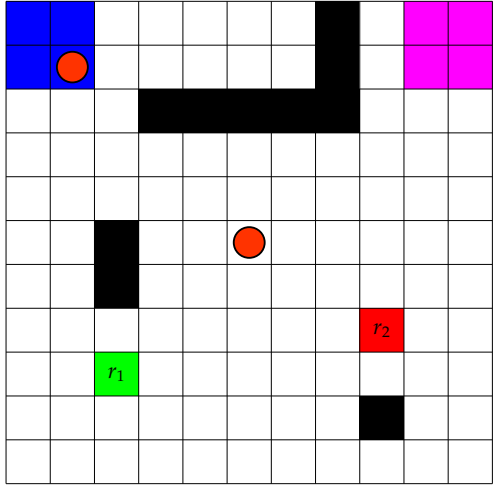

\section{Conclusion}

In this paper, we showed how to compute \newterm{risk-averse} policies. A system governed by such a policy works towards the satisfaction of some given $\omega$-regular specification even in probabilistic environments in which almost sure non-satisfaction of the specification cannot be avoided in the long run. Instead of just resigning because the probability mass of the runs of a Markov decision process that satisfy the specification can only be $0$, a $p$-risk averse policy always reaches the respective next \emph{goal state} with a probability of at least $p$ (from the previous goal state). The definition of the problem ensures that the goal states are chosen in a way that faithfully captures the satisfaction of the specification. We assumed that the specification is given as a deterministic parity automaton, but structured logics such as linear time logic (LTL) could also be used, as translations from LTL to parity automata are known \cite{DBLP:journals/lmcs/Piterman07}.\looseness=-1

We intent to extend the approach to the synthesis of strategies in stochastic two-player games in future work. Also, we will explore how to incorporate additional optimization criteria such as mean-average cost into policy generation and if reinforcement learning techniques can be used to successively approximate optimal policies during policy execution.

\section*{Acknowledgements}

\noindent R.~Ehlers was supported by the Institutional Strategy of the University of Bremen, funded by the German Excellence Initiative.
S. Moarref and U. Topcu were partially supported by awards AFRL $FA8650$-$15$-$C$-$2546$, ONR $N000141310778$, ARO $W911NF$-$15$-$1$-$0592$, NSF $1550212$, and DARPA $W911NF$-$16$-$1$-$0001$.

\bibliographystyle{IEEEtran}
\bibliography{bib}

\end{document}